\documentclass[journal]{IEEEtran}
\usepackage{graphicx}
\usepackage{amsmath}
\usepackage{amsthm}
\usepackage{amsfonts,amssymb}
\usepackage{mathtools}
\usepackage{algorithm}
\usepackage{algorithmic}
\usepackage{makecell}
\usepackage[numbers,sort&compress]{natbib}
\usepackage{subfigure}
\newtheorem{theorem}{Theorem}
\newtheorem{definition}{Definition}

\ifCLASSINFOpdf
\else
\fi
\begin{document}
\title{Data Privacy and Utility Trade-Off Based on Mutual Information Neural Estimator}

\author{Qihong~Wu, Jinchuan~Tang,~\IEEEmembership{Member,~IEEE,}
	Shuping Dang,~\IEEEmembership{Member,~IEEE,}  Gaojie~Chen,~\IEEEmembership{Senior~Member,~IEEE}
\thanks{This work was supported by Guizhou University through ``The Secure Encryption Mechanisms of Spatially Embedded Networks'' under Grant No. 702957213301.}
\thanks{Q. Wu and J. Tang is with the School of Computer Science and Technology, Guizhou University, Guiyang, P.R. China e-mail: \{gs.qhwu, jctang\}@gzu.edu.cn. J. Tang is the corresponding author.} 
\thanks{S. Dang is with the Department of Electrical \& Electronic Engineering, University of Bristol, Bristol, UK e-mail: shuping.dang@bristol.ac.uk.}
\thanks{G. Chen is with the Department of Engineering, University of Leicester, Leicester, UK e-mail: gaojie.chen@leicester.ac.uk.}
}
\maketitle

\begin{abstract}
In the era of big data and the Internet of Things (IoT), data owners need to share a large amount of data with the intended receivers in an insecure environment, posing a trade-off issue between user privacy and data utility. The privacy utility trade-off was facilitated through a privacy funnel based on mutual information. Nevertheless, it is challenging to characterize the mutual information accurately with small sample size or unknown distribution functions. In this article, we propose a privacy funnel based on mutual information neural estimator (MINE) to optimize the privacy utility trade-off by estimating mutual information. Instead of computing mutual information in traditional way, we estimate it using an MINE, which obtains the estimated mutual information in a trained way, ensuring that the estimation results are as precise as possible. We employ estimated mutual information as a measure of privacy and utility, and then form a problem to optimize data utility by training a neural network while the estimator's privacy discourse is less than a threshold. The simulation results also demonstrated that the estimated mutual information from MINE works very well to approximate the mutual information even with a limited number of samples to quantify privacy leakage and data utility retention, as well as optimize the privacy utility trade-off.
\end{abstract}


% \fi
\IEEEpeerreviewmaketitle

\section{Introduction}
\IEEEPARstart{D}{ata} privacy involves publishing data efficiently to minimize risk and protect sensitive data. The content of data and its associated metadata information can be exploited to infer sensitive and personally identifiable information, which can be harmful to individuals and organizations. The availability of data publishing are credited with advancing solutions to complex problems in data sharing, data acquisition, among others in the era of big data and IoT \cite{8462776, article1}. However, big data publishing comes with a massive privacy disclosure where explosive sensitive data growth has been witnessed. The authors of \cite{1067} have already demonstrated the impact of individual data loss by using examples of privacy disclosure and listed the complexity of tackling technological and legislative challenges for big data and individual privacy in the age of IoT.

Usually, the approaches to keep certain information private can be achieved by distorting the information while disclosing relevant information \cite{1056749,10.1145/342009.335438,dwork2006calibrating,10.1007/11787006_1,6482222}. The authors in \cite{971193, 10.1142/S0218488502001648} established \textit{k}-anonymity  as the characteristic that each record is indistinguishable from at least $k-1$ other records on the quasi-identification to prevent the identity of the owner of public data from being revealed. Generalization and suppression were utilized to obtain \textit{k}-anonymity. Although \textit{k}-anonymity overcomes the problem of identity disclosure, it does not preclude attribute disclosure or homogeneity attacks. Individuals may be exposed to the relationship between identities and sensitive attributes, which may jeopardize the distribution of individuals and the entire dataset. To address this issue, the authors of \cite{10.1145/1217299.1217302} introduced $l$-diversity, which requires sensitive characteristics to have at least $l$ well-represented values in each equivalence class to withstand the homogeneity attack in $k$-anonymity. As noted in \cite{10.1145/1217299.1217302}, $l$-diversity is faced with two major attacks: the first is similarity attack, which ignores the danger posed by semantic relationships between attributes; the second is skewness attack, in which the adversary might deduce sensitive information based on the distribution of sensitive qualities, which is a serious invasion of privacy. As a result, preventing property disclosure issues is insufficient. To this end, the $t$-closeness approach was introduced in \cite{4221659} where the earth mover's distance was used to compute the $t$-closeness. But it has drawbacks such as large loss of data utility and inability to discriminate semantic information. 
The term \textit{utility} refers to certain system properties and intelligibility which represents the amount of useful information that can be extracted from the protected data. 

In an attempt to provide better data utility, the authors of \cite{10.1007/s00778-014-0351-4} developed a model which combines $k$-anonymity and differential privacy, where the latter is capable of answering queries to statistical databases with provable privacy guarantees by minimizing attackers' opportunities to identify database records \cite{dwork2006calibrating, 10.1007/11787006_1}. The requirement that information is disseminated only to a limited extent while the data still meets the usability of certain desired recipients has becomes critical nowadays as a large amount of personal information is being widely disseminated, shared and openly accessible by anyone in big data and IoT applications. Since distorting data too much will destroy the value of data to the desired receipents while distorting too little will help the adversaries to deduce the sensitive information of targeted individual, finding an acceptable privacy utility trade-off between privacy protection and data utility is the key in the development and deployment of privacy protection methods. To make things work, the measurements and characterizations on both the privacy and utility have to be achieved. 

The privacy protection, data uitility as well as their trade-off can be characterized mathematically with the help of information theory. In \cite{1056749, 6482222}, the authors consider expected distortion as a measure of privacy and utility where the collective privacy of all or subsets of database items are obtained to offer a progressive conclusions of the fuzzy region of rate distortion as the number of data samples rises arbitrarily. To simulate privacy leakage and data distortion, the authors of \cite{6970882} proposed a concept called the ``privacy funnel" to characterize the trade-off between data privacy disclosure and utility. In simple terms, the privacy funnel used mutual information to measure both data distortion and privacy disclosure. However, their proposed greedy algorithm used merged elements in the context of unknown distribution of data sets to approximate mutual information which runs fast at the expense of accuracy. Precise calculations on mutual information only apply to discrete variables because the sum can be precisely calculated, or to finite problems where the probability distribution is known. For more general problems, this is impossible. 
Furthermore, to process complex data when the data distribution is unavailable, the authors of \cite{DBLP:journals/corr/ZhangOSO17} added differential privacy guarantee noise to the data compressed by the autoencoder to resist inference attacks in motion-aware applications. Although differential privacy provided strong protection for each identified entry, it failed to separate the information associated with sensitive data from the information associated with non-sensitive data. Thus, the increased noise level seriously compromised data utility. Meanwhile, the authors of \cite{DBLP:conf/aaai/XuCSMS17} proposed a model where the private information is equal to zero when the first stage of the expected information predictor is a linear operator. Although the model achieves a perfect utility competitive trade-off when utility and private data are orthogonal, it requires both an understanding of the model used to handle downstream utility tasks  and the first part of the model must be linear. However, in data publishing, it is usually impossible to obtain the information about such downstream task models. This encourages us to estimate mutual information using a data-driven approach and optimize privacy funnel using mutual information as a measure of privacy and utility.

Recently, a mutual information neural estimator (MINE) was proposed in \cite{Belghazi2018MutualIN} to obtain the approximated mutual information from data samples. Unlike traditional non-universal methods based on the partitioning of probability space \cite{Fraser1986IndependentCF, 761290}, \textit{k}-nearest neighbor statistics \cite{PhysRevE.69.066138, 8294268}, maximum likelihood estimation \cite{DBLP:journals/jmlr/SuzukiSSK08} and variational lower bound \cite{Barber2003TheIA}, MINE does not require the data distribution that is difficult to obtain in practice and it can also match the sample size and dimension better. The simulation results show that such a method can estimate mutual information through neural network training in the case of unknown data distribution, and it is almost indistinguishable from real mutual information. The authors of \cite{9022174} applied MINE to train hidden layer with information bottleneck loss, and froze it before moving onto the next layer. The authors of \cite{Qian2021LearningRV} relaxed the Gaussian posterior assumption by using the MINE \cite{Belghazi2018MutualIN} to train a variational information bottleneck to improves its robustness to adversarial perturbations. The work in \cite{hjelm2019learning} investigated unsupervised learning of representations by maximizing mutual information based on MINE between an input and the output of a deep neural network encoder, and MINE has been shown once again to estimate mutual information more precisely than traditional methods. This inspires us to use mutual information as a measure of both privacy and utility, which will be formulated as optimization problems for privacy utility trade-off.

In this paper, mutual information is adopted to measure both privacy discourse and data utility at the same time. %Among some practical issues, it is difficult to obtain mutual information of real data sets. 
Different from previous works which used merged elements to approach data distribution for privacy and utility characterizations, this paper uses MINE \cite{Belghazi2018MutualIN} instead for accuracy. Then, by interval training mutual information estimator and encoder to maximize mutual information, we put mutual information into the privacy utility trade-off framework to obtain the privacy utility model using only single mutual information to measure data privacy and utility, making the model more flexible. Simulations demonstrate that the model's practicality and assess the impact of various noise levels and privacy precisions on the maximum data utility. The simple model allows us to take different parameters on the model into consideration, which includes the influence of different noise levels on data utility and the influence of privacy accuracy on the overall model.

This contribution can be summarized as follows:
\begin{itemize}
	\item The mutual information estimator is used to safeguard data privacy when it is published. It overcomes the problem of unknown distributions being unable to extract mutual information and maximizes data utility under the assumption of fixed privacy budget.
	\item Unlike other measuring methods, the privacy utility trade-off which we proposed relies solely on mutual information. MINE is used throughout this work to estimate mutual information for a better accuracy. And the proposed trade-off framework is straightforward and simple to implement. 
	\item We perform the simulation to find the impacts of various parameters such as noise levels and privacy budgets on the privacy utility trade-off framework and offer some useful insights on choosing the parameters.
\end{itemize}

The notation is given in Table \ref{Table.1}. The remainder of the paper is laid out as follows. In Section \ref{Sec.2}, we describe the system model. Then we formulate the privacy utility trade-off objective using neural estimator in Section \ref{Sec.3}, and Section \ref{Sec.4} is extensive simulation part, finally we summarize our comments in Section \ref{Sec.5}.

\begin{table}
	\centering
	\caption{Notations}
	\def\arraystretch{1.5}
	\begin{tabular}{ c c c c c c c c c c c c c c c c c c c c c c }
		\hline
		Symbol & Meaning  \\ 
		\hline
		$C(S,q)$ & 
		Cost function\\
		$d(X,Y)$ &	 Distortion measure\\
		$D$ &	  Distortion level\\
		$D_{KL} (\cdot||\cdot)$ &	 KL-divergence of two distributions\\
		$E_{P_{X,Y}}[d(X,Y)]$ & 	 Expected distortion\\
		$\hat{G}$ & Estimated gradient \\
		$H(\cdot)$ &  Information entropy\\
		$\hat{H}_{\theta}(\cdot)$ &  Estimated cross-entropy with neural network $\theta$\\
		$I(\cdot;\cdot)$ &  Mutual information of two random variables\\
		$\hat{I}_\theta (\cdot;\cdot)$ &	Estimated mutual information with neural network $\theta$ \\
		$k$ & Minibatch size\\
		$n$ &  Number of samples from $Y$\\
		$p(\cdot,\cdot)$ &  Joint probability density of two random variables\\
		$p(\cdot)$ &  Marginal probability density\\
		$\hat{P}^{\left(n\right)}$  & \begin{tabular}[c]{@{}l@{}}
			Empirical distribution associated to $n$ i.i.d. \\ samples for a given distribution $P$\\
		\end{tabular}\\
	   $P(Y|X)$ &	Conditional distribution, privacy mapping\\
	   $q$	&  Belief distribution\\
	   $q_0^*$ &	 Optimal $q$ before observation of data\\
	   $q_y^*$ & 	After observing published data\\
		$S$ &  Sensitive dataset  \\
		$T_{\epsilon}(X;Y)$ & 	 Characterization of privacy utility trade-off\\
		$X$ &  Dataset related to S\\  
		$Y$ & Published dataset \\
		$\Delta C$ &	Expected inference cost gain\\
		$\epsilon$ & 	Privacy parameter, privacy budget\\
		$\eta$ & Learning rate\\
		$\theta$ &	 Deep neural network parameter\\
		$\phi$ & 	 Encoder weights\\ 
		\hline
		\label{Table.1}
	\end{tabular}
\end{table}

\section{System Model}
\label{Sec.2}
As shown in Fig. \ref{Fig.1}, we consider a case where a data publisher has some sensitive data set $S \in \mathcal{S}$ which is correlated with some non-sensitive data set $X \in \mathcal{X}$, and the data publisher wishes to share $X$ with data receiver. This correlation can be regarded as an auxiliary information to infer the sensitive dataset $S$. To reduce the disclosure of sensitive dataset $S$, the data publisher send a perturbed version of $X$ denoted by $Y \in \mathcal{Y}$, where $Y$ is the published dataset generated by adding Gaussian or Laplacian noise to $X$. Dataset $Y$ that passes through the privacy mechanism should be as irrelevant as possible to sensitive data $S$, while retaining as much information as possible about $X$, since $Y$ will be processed by the data receiver to provide utility. We assume that $S, X$ and $Y$ follow a Markov chain $S \rightarrow X \rightarrow Y$. 

\begin{figure}[t]
	\centering
	\includegraphics[width=0.5\textwidth]{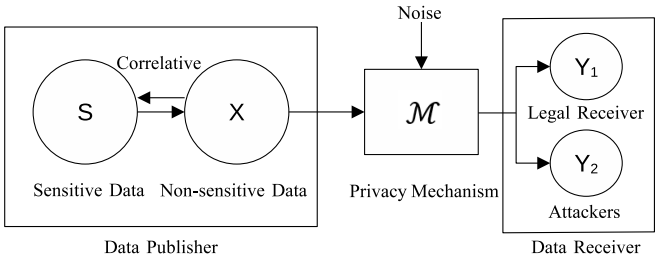}
	\caption{Privacy mechanism model.}
	\label{Fig.1}	
\end{figure}	

Under the assumptions above, we propose to design a privacy framework which maps $X$ to a random variable $Y$ so that the information leakage from $S$ to $Y$ is limited while the data utility of $X$ given $Y$ is maximum. Our challenge, therefore, is to use a data-driven approach to find the privacy utility trade-off for both discrete and continuous data. To fulfil such vision, unlike \cite{1056749}, which used a greedy algorithm with merged elements, we propose to use the mutual information nerual network estimator to obtain the more accurate metrics of information leakage from $S$ to $Y$ as well as the utility of $X$ given $Y$. The combined results on privacy and utility measures will be used in the proposed privacy utility trade-off framework to optimize utility while safeguarding privacy.

\section{Privacy utility trade-off}
\label{Sec.3}

In this section, we describe the privacy and utility metrics in data publishing separately first. Then, we introduce the estimator for mutual information. Finally, we discuss the privacy utility trade-off problem which firstly adopts mutual information neural estimator as a measure of both privacy and utility.

\subsection{Privacy and Utility Metric}
Here, we consider the inference threat model introduced in \cite{DBLP:conf/allerton/CalmonF12}. At this moment, illegal data receiver as an attacker tries to infer sensitive information $S$. More precisely, the data receiver chooses a belief distribution $q$ from the set $P_S$ of all probability distributions over $S$ to minimize the expected inference cost function $C(S,q)$. Under the logarithmic loss $C(S,q)=-\log_2 q(S)$, the optimal $q$ before observation of data $Y$ is
\begin{equation}
	\begin{aligned}
		q_0^* &= \mathop{\arg \min} \limits_{q \in P_S}E_{P_S}[C(S,q)]\\
		&=\mathop{\arg \min} \limits_{q \in P_S}E_{P_S}[-\log_2 q(S)],
	\end{aligned}
\end{equation}
and the Shannon entropy of $S$ is given by
\begin{equation}
	\begin{aligned}
		H(S) &=E_{P_S}[-\log_2 q_0^*(S)].\\
	\end{aligned}
\end{equation}
After observing published data $y \in Y$, we get
\begin{equation}
	\begin{aligned}
 		q_y^* &= \mathop{\arg \min} \limits_{q \in P_S}E_{P_{S|Y}}[C(S,q)|Y=y]\\
		&=\mathop{\arg \min} \limits_{q \in P_S}E_{P_{S|Y}}[-\log_2 q(S|Y)],\\
	\end{aligned}
\end{equation}
and the conditional entropy of $S$ given $Y$ is written as
\begin{equation}
	\begin{aligned}
		H(S|Y) &=E_{P_{S|Y}}[-\log_2 q_y^*(S|Y)].
	\end{aligned}
\end{equation}

Because publishing data $Y$ is observable, so the data receiver obtains an expected gain in inference cost of \begin{equation}
	\begin{aligned}
\Delta C=H(S)-H(S|Y)=I(S;Y),
	\end{aligned}
\end{equation}
where the expected inference cost gain $\Delta C$ measures the improvement in the inference quality from sensitive data $S$ after the observed $Y$. The design goal of the privacy utility trade-off should aim at reducing $\Delta C$ to bring the inference with observing $Y$ closer to the initial inference without observing $Y$, the formula can be written as
\begin{equation}
	\begin{aligned}
	\min \Delta C &= \min I(S;Y)\\
	&=\min \int_{S \times Y}p(s,y)\log_2{\frac{p(s,y)}{{p(s)p(y)}} }{\rm d}s\,{\rm d}y.
	\label{eq_delta}
\end{aligned}
\end{equation}
where $p(s,y)$ is the joint probability density of $S$ and $Y$, $p(s)$ and $p(y)$ are the corresponding marginal probability density.

Protecting the privacy of data publishing is considered essential, at the same time, the framework of privacy utility trade-off should maintain the utility of the perturbed data $Y$. In \cite{book}, the key issue is that utility is determined not only by the data modifications made, but also by the anticipated data uses. Since potential data uses are truly diverse and it may even be hard to identify them all now of the data release, privacy protection can seldom be performed in a data use-specific manner. As a result, it is more common to speak about information loss rather than utility. Measures of information loss provide a basic approach for the data protector to determine how much harm a particular masking technique is causing to the data.

As the discrimination between two distribution of mean information, KL-divergence can be used to measure information loss, which formula is as \eqref{Eq.2}. Meanwhile, we know that KL-divergence is equal to mutual information by one form as \eqref{Eq.3}. Consequently, we can measure data utility by mutual information.

\subsection{Mutual Information Neural Estimator (MINE)}
The theoretical mutual information defined above depends on probability density functions which are difficult to obtain in practice. Here, we focus on an estimator named MINE for its good performance in approximating the mutual information from data samples with high accuracy \cite{Belghazi2018MutualIN}. It utilizes the Donsker-Varadhan representation (DV-representation) of the Kullback-Leibler divergence (KL-divergence), which is related to the mutual information. In MINE, mutual information is estimated by parameterizing the lower bound of KL-divergence and improving the lower bound by continuous training.

The KL-divergence is a measurement of the dissimilarity between distributions $P$ and $Q$, which can be written as
\begin{equation}
	D_{KL}(P||Q) = \int{P(x)\log_2{\frac{P(x)}{Q(x)}}}{\rm d}x,
	\label{Eq.2}
\end{equation}
where $D_{KL}(P||Q)\lg 0$, with $D_{KL}(P||Q) = 0$ if and only if $P(x) = Q(x)$. 
Consequently, we get
\begin{equation}
	\begin{aligned}
		I(S;Y)&:=\int_{S\times Y} p(s,y)\log_2{\frac{p(s,y)}{{p(s)p(y)}} }{\rm d}s\, {\rm d}y \\
		&=D_{KL} (p(s,y)||p(s)p(y))\\
		&=E_{p(S,Y)}\left[\log_2 \left(\frac{p(s,y)}{p(s)p(y)}\right)\right],
		\label{Eq.3}
	\end{aligned}
\end{equation}
where, given samples of $S$ and $Y$, we can estimate $I(S;Y)$ at the cost of accuracy when the number of samples disobey the Law of Large Numbers. In this work, we pay attention to DV-representation \cite{https://doi.org/10.1002/cpa.3160360204}, which could result in a more accurate estimator.
\begin{theorem}
	(DV-representation). The KL-divergence admits representation given by 	\begin{equation}
		D_{KL}(P||Q) = \sup_{T: \Omega \rightarrow R} E_P [T] - \log_2 (E_Q[e^T]),
	\end{equation}
where the supremum is taken over all functions $T$ such that the two expectations are finite.
\end{theorem}

\begin{proof}
	See Appendix. 
\end{proof}

Let $\mathcal{F}$ be any class of functions $T:\Omega \rightarrow R$ satisfying the integrability constraints of the theorem, we then have a lower-bound given by
\begin{equation}
D_{KL}(P||Q) \ge \sup_{T\in \mathcal{F}} E_P [T] - \log_2 (E_Q[e^T]).
	\label{Eq.11}
\end{equation}

For an MINE, the approach here is to choose $\mathcal{F}$ to be the family of functions $T_\theta: \mathcal{S} \times \mathcal{Y} \rightarrow\mathbb{R} $ parametrized by a deep neural network with parameter $\theta \in \Theta$ so that
\begin{equation}
	I(S;Y) \ge I_\Theta (S;Y) = \sup_{\theta \in \Theta} E_{P_{SY}}[T_\theta]-\log_2 \left(E_{P_S P_Y}[e^{T_\theta}]\right),
\end{equation}
the expectations above are estimated using empirical samples from $P_{SY}$ and $P_S \times P_Y$ or by shuffling the samples from the joint distribution along the batch axis, which is defined by Definition \ref{def1}. As a result, the MINE algorithm to train the estimator is given by Algorithm \ref{alg1}, and the structure of a four-layer structured of neural network, as shown in Fig. \ref{Fig.2}, is adopted in this work. 
\begin{definition} \label{def1}
	(Mutual Information Neural Estimator (MINE)). Let $ \mathcal{F} = \left\{T_\theta\right\}_{\theta \in \Theta} $ the set of functions parametrized by a neural network. MINE is defined as \cite{Belghazi2018MutualIN}	
	\begin{small}
		\begin{equation}
			{\hat{I}(S;Y)_n}=\sup_{\theta \in \Theta}{E_{P_{SY}^{ \left ( n\right )}}}[T_\theta]-\log_2\left({E_{P_S^{\left(n\right)}\hat{P}_Y^{\left(n\right)}}}[e^{T_\theta}]\right),
		\end{equation}
	\end{small}
where $n$ is the number of samples, $\hat{P}^{\left(n\right)} $ is the empirical distribution associated to $n$ independent and identically distributed (i.i.d.) samples by given distributions $P_{SY}, P_{S}, P_{Y}$. 	
\end{definition}

\begin{figure}[t]
	\centering
	\includegraphics[width=0.5\textwidth]{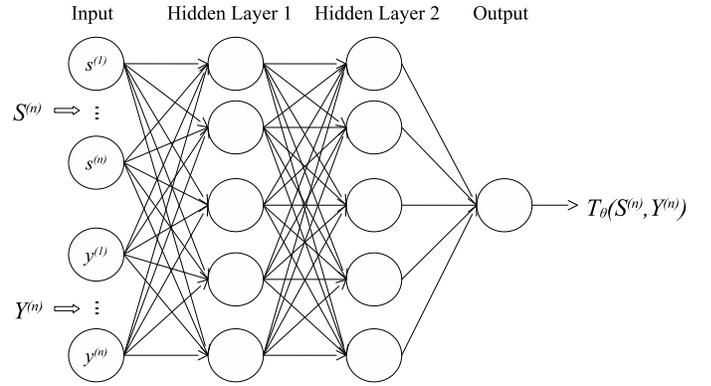}
	\caption{A four-layer neural network approximated $T_\theta$.}
	\label{Fig.2}	
\end{figure}

\begin{algorithm}[t]
	%\textsl{}\setstretch{1.8}
	\renewcommand{\algorithmicrequire}{\textbf{Input:}}
	\renewcommand{\algorithmicensure}{\textbf{Output:}}
	\caption{Mutual information neural estimator \cite{Belghazi2018MutualIN}}
	\label{alg1}
	\begin{algorithmic}[1]
		\REQUIRE $\left(s^{\left(1\right)}, y^{\left(1\right)}\right), ..., \left(s^{\left(k\right)}, y^{\left(k\right)}\right)$, $k$, $\eta$, where $(s,y)$ is a given dataset, $k$ is minibatch size, and $\eta$ is learning rate.
		\ENSURE $\hat{I}_\theta(S;Y)$
		\REPEAT
		\STATE Draw $k$ minibatch samples from the joint distribution: $\left(s^{\left(1\right)}, y^{\left(1\right)}\right), ..., \left(s^{\left(k\right)}, y^{\left(k\right)}\right)\sim P_{SY} $
		\STATE Draw $n$ samples from the $Y$ marginal distribution: $\overline{y}^{\left(1\right)}, ..., \overline{y}^{\left(n\right)} \sim P_Y$
		\STATE Evaluate the mutual information lower-bound:
	%	\begin{small}
			\begin{align*}
				\hat{I}_\theta(S;Y) \leftarrow & {\frac{1}{k}{\sum_{i=1}^{k}T_\theta {\left(s^{\left(i\right)},y^{\left(i\right)}\right)}}}\\&-\log_2{\left(\frac{1}{k}\sum_{i=1}^{k}e^{T_\theta {\left(s^{\left(i\right)},y^{\left(i\right)}\right)}}\right)}
			\end{align*}
	%	\end{small}
		\STATE Evaluate bias corrected gradients with moving average: $\hat{G}(\theta) \leftarrow \widetilde{\nabla}_\theta \,\hat{I}_\theta(S;Y)$
		\STATE Update the model parameters: $\theta \leftarrow \theta + \eta \hat{G}(\theta)$
		\UNTIL convergence
	\end{algorithmic}  
\end{algorithm}

%In Algorithm \ref{alg1}, the gradient estimate:
%\begin{equation}
%	\hat{G}(\theta) = E[\nabla_\theta T_\theta] - \frac{E[\nabla_\theta T_\theta e^{T_\theta}]}{E[e^{T_\theta}]},
%\end{equation}
%where $\nabla_\theta$ is gradient, the expectations are over the samples of a minibatch $k$, leads to a biased estimate of the full batch gradient. So, we will reduce the bias by replacing the estimate in the denominator by an exponential moving average.

\subsection{Privacy Utility Trade-off}
To model the privacy disclosure and data utility, the authors of \cite{6970882} proposed a concept called the “privacy funnel" to represent the trade-off between data utility and user privacy. Briefly, the privacy funnel uses mutual information to evaluate both data utility and privacy disclosure. And the optimization problem of privacy utility trade-off for a given distortion level $D$ is characterized as follow:
\begin{equation}
	\min_{P_{Y|X}: E_{P_{X,Y}}[d(X,Y)]\le D}\Delta C,
	\label{12}
\end{equation}
where $P_{Y|X}$ is privacy mapping, $d(X,Y)$ is the distortion measure, $E_{P_{X,Y}}[d(X,Y)]$ represents the expected distortion, $D$ represents the level of distortion, $\Delta C$ describes expected inference cost gain. Since the data distribution is often inaccessible and only a proportion of the sampled data is available, we propose to use MINE to obtain estimated mutual information as a measure of both privacy and utility. 

The dual form of \eqref{12} is given by
\begin{equation}
	T_\epsilon(X;Y):=\max_{\substack{P_{Y|X}:\hat{I}(S;Y) \le \epsilon\\S \rightarrow X \rightarrow Y}} \hat{I}(X;Y),
	\label{Eq.13}
\end{equation}
where $T_\epsilon (X;Y)$ is characterization of privacy utility trade-off; $P_{Y|X}$ represents privacy mapping; $\hat{I}(X;Y)$ and $\hat{I}(S;Y)$ is the estimated mutual information of $I(X;Y)$ and $I(S;Y)$, respectively; $\epsilon$ stands for privacy budget.

The privacy utility trade-off with a privacy budget can be used as a solution to the optimization problem of \eqref{Eq.13}. The remained task is to design a trade-off algorithm. From a communication theoretic perspective, privacy mechanism is similar to a noisy channel, and the optimal transmission rate can be regarded as a function of the mutual information $I(X;Y)$ between input $X$ and output $Y$ of a channel $P_{Y|X}$. However, the mutual information also depends on the channel probability distribution, i.e. added noise. Rather than approximating the channel probability distribution itself, we will approximate the mutual information $I(X;Y)$ between the samples of the channel input and output. Then we optimize the mutual information due to added noise so that the adversaries will be diffcult to infer the sensitive data, which can be regarded as an encoding process, and also optimize the infering capability of legal receiver by minimizing cross-entropy mutual information between $Y$ and $X$ so that the data utility is guaranteed, which can be regarded as a decoding process. Thus, we give our optimization framework as Algorithm \ref{alg2}.

\begin{algorithm}[h]\label{alg2}
	%\textsl{}\setstretch{1.8}
	\renewcommand{\algorithmicrequire}{\textbf{Input:}}
	\renewcommand{\algorithmicensure}{\textbf{Output:}}
	\caption{Privacy Utility Trade-off}
	\label{alg2}
	\begin{algorithmic}[1]
		\REQUIRE $S, x, y, n_0, \epsilon, \eta, \theta_1, \theta_2$ and $\theta_3$, where $S$ is sensitive datasets, $x,y$ are data samples, $n_0$ denotes additive nosie, $\epsilon$ is privacy budget; $\eta$ is learning rate, $\theta_1$ is the network for estimation training process, $\theta_2$ is the network for an encoder tuned by $\theta_1$, and $\theta_3$ is the decoding network.
		\ENSURE $\hat{I}_{\theta_1}(X;Y), \hat{I}_{\theta_2}(X;Y), \hat{H}_{\theta_3}(x), \hat{x}$
		\STATE $y \leftarrow x + n_0$
		\STATE Calculate the approximated $I(S;Y)$ with Algorithm \ref{alg1}
		\WHILE {$\hat{I}(S;Y) \le \epsilon$}
		\STATE \textbf{Encoding training process}:
		\REPEAT
		\STATE Evaluate the estimated mutual information lower-bound: 
	%	\begin{small}
			%\begin{equation*}
				\begin{align*}
				\hat{I}_{\theta_1}(X;Y) \leftarrow & {\frac{1}{k}{\sum_{i=1}^{k}T_{\theta_1} {\left(x^{\left(i\right)},y^{\left(i\right)}\right)}}}\\&-\log_2{\left(\frac{1}{k}\sum_{i=1}^{k}e^{T_{\theta_1} {\left(x^{\left(i\right)},y^{\left(i\right)}\right)}}\right)}
			\end{align*}
			%\end{equation*}
	%	\end{small}
		
		\STATE To maximize $\hat{I}_{\theta_1}(X_{\theta_2};Y) $ in the encoding process,  gradients:\\
		$\hat{G}_{xy}(\theta_1) \leftarrow \widetilde{\nabla}_{\theta_1} \,\hat{I}_{\theta_1}(X_{\theta_2};Y)$\\
		$\hat{G}_{xy}(\theta_2) \leftarrow \widetilde{\nabla}_{\theta_2} \,\hat{I}_{\theta_1}(X_{\theta_2};Y)$,\\ 
		where $X_{\theta_2}$ is the output of neural network  $\theta_2$.  
		\STATE Update the model parameters:\\
		$\theta_1 \leftarrow \theta_1 + \eta \hat{G}(\theta_1)$\\
		$\theta_2 \leftarrow \theta_2 + \eta \hat{G}(\theta_2)$
		\UNTIL convergence\\
		\STATE \textbf{Decoder training process}:
		\REPEAT 
		\STATE Evaluate the cross-entropy of estimated $\hat {x}$: 
		$$\hat{H}_{\theta_3}(X, \hat{X}) \leftarrow -\frac{1}{k}{\sum_{i=1}^{k}T_{\theta_3} {\left(x^{\left(i\right)}\right)}}\log_2 T_{\theta_3}{\left(x^{\left(i\right)}\right)}$$
		\STATE Evaluate bias corrected gradients with moving average:  $\hat{G}_x(\theta_3) \leftarrow \widetilde{\nabla}_{\theta_3} \hat{H}_{\theta_3}(X)$
		\STATE Update the model parameters: $\theta_3 \leftarrow \theta_3 + \eta \hat{G}(\theta_3)$
		\UNTIL convergence
		\ENDWHILE
	\end{algorithmic}  
\end{algorithm}

\section{Simulation Results}
\label{Sec.4}
In this section, we conduct simulations to verify and evaluate the performance of our proposed model. We first compare true mutual information with estimated mutual information based on MINE for a known data distribution. Then, the proposed privacy utility trade-off framework is assessed from a variety of angles, including the influence of the privacy utility trade-off under various noise levels and privacy budget $\epsilon$. Finally, the effect of noise parameters is taken into account.

\subsection{The Accuracy of the Estimated Mutual Information}
We compare the mutual information with and without MINE. The estimated results are obtained by Algorithm \ref{alg1}, and the true MI is calculated using a mutual information formula. In the simulations, $S$ is a dataset of random numbers that follow a standard normal distribution. $S$ and $Y$ follow a joint normal distribution with zero mean and a variance $0.2$. The true mutual information with \eqref{Eq.3} is $0.6586$, and the maximum epochs is $500$, then we import the data into MINE for iterative estimation. The parameters of the neural network are described in Table \ref{Table.2}. As shown in Fig. \ref{Fig.4}, the estimated mutual information converges to the actual mutual information as the number of training epochs increase. At epoch $250$, the mutual information obtained by MINE is indistinguishable from the true mutual information. 

\begin{table}
	\centering
	\caption{Simulation Parameters}
	\def\arraystretch{1.5}
	\begin{tabular}{ c c c c c c}
		\hline
		Parameter & Symbol & Value \\ 
		\hline
		Learning rate & $\eta$ & 0.0005 \\  
		Neural network layers & $L$ & 3  \\
		Number of learning sessions & epoch & 500 \\  
		Privacy budget & $\epsilon$ & [0.5, 0.75, 1.0, 1.25]\\
		Small batch data & minibatch & 20000  \\
		\hline
		\label{Table.2}
	\end{tabular}
\end{table}

\begin{figure}[t]
	\vspace{-0.7cm}
	\centering
	\includegraphics[width=0.5\textwidth]{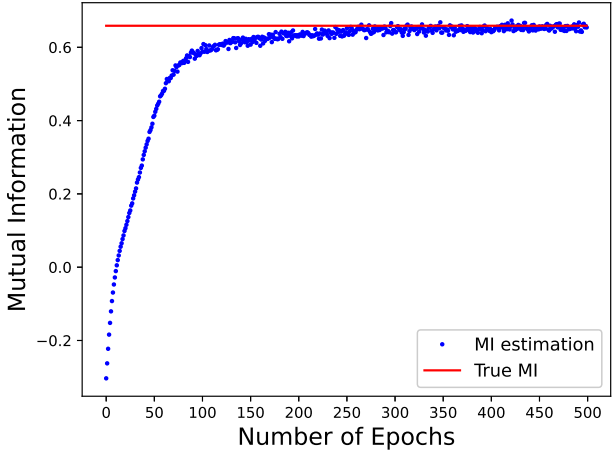}
	\caption{Mutual information estimation and true mutual information.}
	\label{Fig.4}	
\end{figure}

\subsection{Trade-Off Performance for Different Privacy Budgets}
\label{Sec4.2}
The privacy utility trade-off under different privacy budgets is explored in this subsection. Batches of datasets $S$ and $X$ are generated from the same distribution. In addition to this, $Y$ is obtained by adding Gaussian noise to $X$. For mutual information estimation, we use the Adam optimizer and a learning rate of $\eta= 0.0005$. Note that we do not have access to the true joint distribution $p(s,y)$, $p(x,y)$ and the marginal distributions $p(s)$, $p(x)$, $p(y)$, We therefore utilize samples of these distributions and approximated the expectations by the sample average.

Here, the values of privacy budget $\epsilon$ are taken as $0.5$, $0.75$, $1$ and $1.25$, respectively. The maximum mutual information $I(X;Y)$ of different privacy budgets is displayed in Fig. \ref{Fig.5}. As shown in Fig. \ref{Fig.5}, the estimated mutual information value of each epoch gradually grows and eventually converges to near a maximum value, indicating that the threshold of privacy disclosure rises within a specific range. The more sensitive information which is published, the more utility recipients receive, and the more stable the data becomes after training.
The greatest value of mutual information collected in the last epochs for varied privacy budgets. Fig. \ref{Fig.6} shows the comparison results between the various precision of privacy disclosure. It demonstrates that a larger privacy budget would lead to both more privacy disclosure and higher data utility.

\begin{figure}[t]
	\centering
	\includegraphics[width=0.5\textwidth]{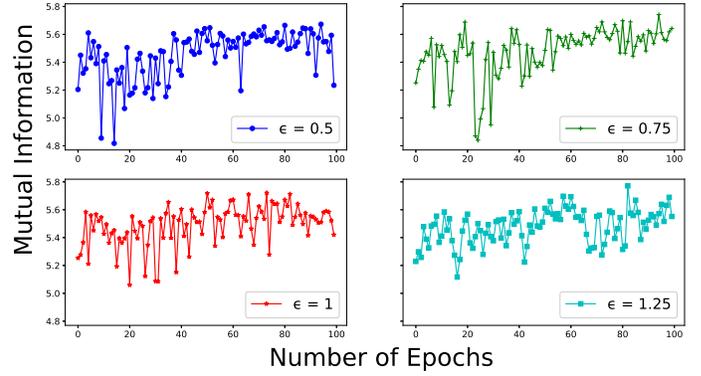}
	\caption{Trade-off performance for different privacy budgets with respect to number of epochs.}
	\label{Fig.5}
\end{figure}

\begin{figure}[t]
	\centering
	\includegraphics[width=0.5\textwidth]{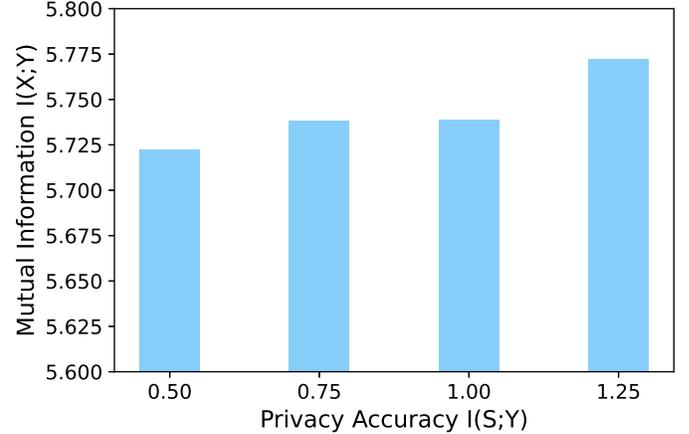}
	\caption{Trade-off performance for privacy budgets $0.5$, $0.75$, $1$ and $1.25$.}
	\label{Fig.6}	
\end{figure}

\subsection{Trade-Off Performance for Different Gaussian and Laplacian Noise Levels}
In this subsection, we consider the influence of adding noise to dataset for different noise levels. We studied privacy utility trade-off using Gaussian noise and Laplacian noise, respectively. In order to assess the influence of different noise levels, we set $\epsilon$ to $0.75$ and the other hyperparameters as in Subsection \ref{Sec4.2}. After that, we choose to add either Gaussian noise or Laplacian noise and perform training for $100$ epochs. We evaluate the privacy utility trade-off at the end of each epoch, and collect the budget of the best utility for different noise levels. As demonstrated in Fig. \ref{Fig.7}, the impact of different noise levels on the privacy utility trade-off varies. It showcases the scenario when the estimated mutual information changes with different noise levels. It shows that the mutual information obtained by adding Gaussian noise is more stable and larger than that obtained by adding Laplacian noise while it tends to converge.

\begin{figure}[t]
	\vspace{-0.52cm}
	\setlength{\abovecaptionskip}{0.4cm}   %调整图片标题与图距离
	\centering
	\includegraphics[width=0.55\textwidth]{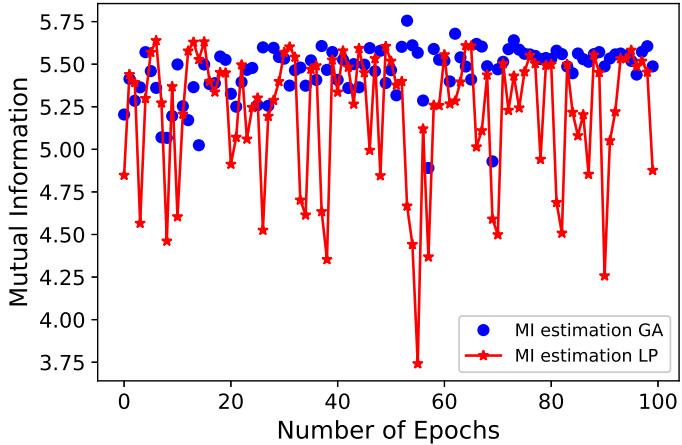}
	\caption{Trade-off performance for different Gaussian (GA) and Laplacian (LP) noise levels.}
	\label{Fig.7}	
\end{figure}

\begin{figure}[t]
	\vspace{-0.55cm}
	\centering
	\includegraphics[width=0.55\textwidth]{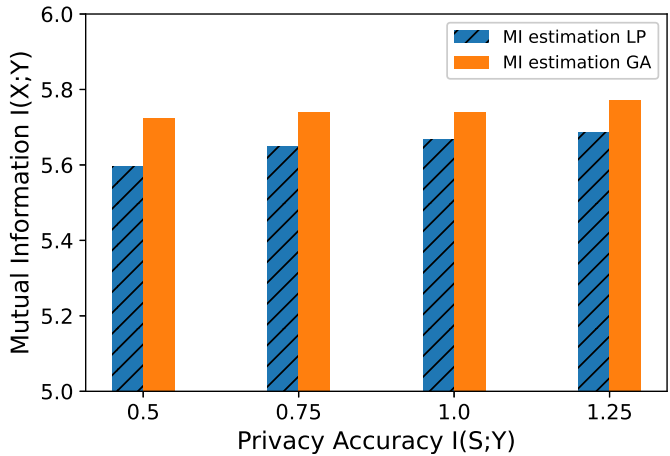}
	\caption{ Maximum utility of adding different noise levels with privacy budgets $0.5$, $0.75$, $1$ and $1.25$.}
	\label{Fig.8}	
\end{figure}

\begin{figure}[t]
%	\vspace{0.25cm}
	\centering
	\includegraphics[width=0.5\textwidth]{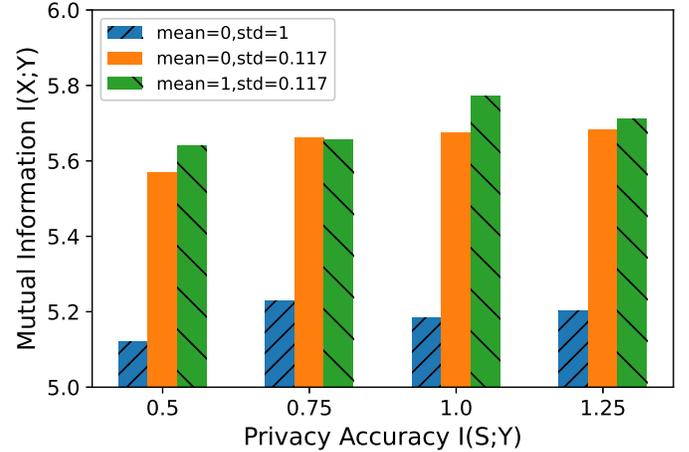}
	\caption{Trade-off performance for different Gaussian parameters with privacy budgets $0.5$, $0.75$, $1$ and $1.25$.}
	\label{Fig.9}	
\end{figure}

Fig. \ref{Fig.8} illustrates the maximum value of mutual information, namely the maximum utility when the privacy budgets are $0.5$, $0.75$, $1.0$ and $1.25$, respectively, by adding Gaussian noise and Laplacian noise to the dataset. It can be clearly observed that when the same privacy budget is taken, the maximum utility obtained by adding Gaussian noise is greater than that obtained by adding Laplacian noise. The maximum data utility rises in tandem with the level of privacy budget.

\subsection{Trade-Off Performance for Different Gaussian Parameters}

In this subsection, we consider the influence of setting different Gaussian parameters for maximum mutual information. The privacy utility trade-off of Gaussian noise with the standard deviation $0.1117$ versus noise generated by standard Gaussian distribution are investigated respectively. The privacy utility trade-off is affected heavily by different noise parameters. The maximum mutual information $I(X;Y)$ as shown in Fig. \ref{Fig.9}.
To evaluate the effects of different Gaussian noise parameters, we set $\epsilon$ to $0.75$ and set the other hyperparameters to the parameters in Subsection \ref{Sec4.2}. Then, we modified the standard deviation in the Gaussian noise from the value $0.1117$ to the deviation of the standard Gaussian distribution, i.e. $1$, and ran the framework for $100$ epochs. We evaluated the privacy utility trade-off at the end of each period, and collect the maximum utility for different parameter settings. It can be seen that different mean values have little influence on maximum utility. However, a smaller standard deviation will lead to greater maximum mutual information, in other words, we can set a small variance to obtain a greater utility.

\section{Conclusion}
\label{Sec.5}
In this paper, we consider the privacy utility trade-off in data publishing, which is very important to IoT and big data applications. Specifically, mutual information was used as a measure of both privacy disclosure and data utility, and MINE is adopted to estimate mutual information for a better accuracy. A privacy utility trade-off framework is proposed by incoperating MINE. Our simple but effective model does not require other training strategies, thus making the model more stable and less sensitive to hyperparameters. Simulation results demonstrated the effectiveness of the framework with regard to different parameter settings.

\appendix[Proof of \eqref{Eq.11}]
For completeness, similar to \cite{Belghazi2018MutualIN}, we prove it as follows.
\begin{proof}
	For a given function T, consider the Gibbs distribution $G(x)=\frac{1}{Z} e^{T(x)}  Q(x)$, where $Z=E_Q [e^{T(x)}]$, then	
	\begin{equation}
		\begin{aligned}
			E_P [T(x)]-\log_2(Z)
			&=E_P [T(x)-\log_2(Z)]\\
			&=E_P \left[\log_2{e^{T(x)}}-\log_2(Z)\right]\\
			&= E_P \left[\log_2 \frac{e^{T(x)}}{Z}\right]\\
			&=E_P\left[\log_2 \frac{e^{T(x)}Q(x)}{ZQ(x)}\right]\\
			&=E_P\left[\log_2 \frac{G(x)}{Q(x)}\right],
		\end{aligned}
	\end{equation}	
	let $\Delta$ be subtract between KL-divergence and above, then
	\begin{equation}
		\begin{aligned}
			\Delta &:= D_{KL}(P||Q)-\left(E_P\left[T(x)\right]-\log_2\left(E_Q \left[e^{T(x)}\right]\right)\right)\\
			&=E_P \left[\log_2 \frac{P(x)}{Q(x)}-\log_2 \frac{G(x)}{Q(x)}\right]\\
			&=E_P\left[\log_2 \frac{P(x)}{G(x)}\right]\\
			&=\int{P(x)\log_2{\frac{P(x)}{Q(x)}}}{\rm d}x\\
			&=D_{KL}(P||G),
		\end{aligned}
	\end{equation}
	because of the non-negative of KL-divergence, $\Delta \ge 0$, so that
	\begin{equation}
		D_{KL}(P||Q) \ge  E_P [T] - \log_2 \left(E_Q[e^T]\right). 
	\end{equation}
Thus, we conclude our proof.
\end{proof}

%\appendices
%\section{Proof of the First Zonklar Equation}\label{DV representation}
%Appendix one text goes here.Appendix \ref{apd2}

%\section{Proof of the First Zonklar Equation}\label{apd2}


%\section*{Acknowledgment}


\ifCLASSOPTIONcaptionsoff
  \newpage
\fi



\bibliographystyle{IEEEtran}
\bibliography{IEEEexample}

\begin{thebibliography}{10}
\providecommand{\url}[1]{#1}
\csname url@samestyle\endcsname
\providecommand{\newblock}{\relax}
\providecommand{\bibinfo}[2]{#2}
\providecommand{\BIBentrySTDinterwordspacing}{\spaceskip=0pt\relax}
\providecommand{\BIBentryALTinterwordstretchfactor}{4}
\providecommand{\BIBentryALTinterwordspacing}{\spaceskip=\fontdimen2\font plus
\BIBentryALTinterwordstretchfactor\fontdimen3\font minus
  \fontdimen4\font\relax}
\providecommand{\BIBforeignlanguage}[2]{{%
\expandafter\ifx\csname l@#1\endcsname\relax
\typeout{** WARNING: IEEEtran.bst: No hyphenation pattern has been}%
\typeout{** loaded for the language `#1'. Using the pattern for}%
\typeout{** the default language instead.}%
\else
\language=\csname l@#1\endcsname
\fi
#2}}
\providecommand{\BIBdecl}{\relax}
\BIBdecl

\bibitem{8462776}
R.~B. Messaoud, N.~Sghaier, M.~A. Moussa, and Y.~Ghamri-Doudane, ``Privacy
  preserving utility-aware mechanism for data uploading phase in participatory
  sensing,'' \emph{IEEE Transactions on Mobile Computing}, vol.~18, no.~9, pp.
  2160--2173, Sept. 2019.

\bibitem{article1}
Ádám {Erdélyi}, T.~{Winkler}, and B.~{Rinner}, ``Privacy protection vs.
  utility in visual data - an objective evaluation framework,''
  \emph{Multimedia Tools and Applications}, vol.~77, no.~2, pp. 2285--2312,
  Jan. 2018.

\bibitem{1067}
M.~Adams, ``Big data and individual privacy in the age of the internet of
  things,'' \emph{Technology Innovation Management Review}, vol.~7, pp. 12--24,
  Apr. 2017.

\bibitem{1056749}
H.~Yamamoto, ``A source coding problem for sources with additional outputs to
  keep secret from the receiver or wiretappers,'' \emph{IEEE Transactions on
  Information Theory}, vol.~29, no.~6, pp. 918--923, Nov. 1983.

\bibitem{10.1145/342009.335438}
R.~Agrawal and R.~Srikant, ``Privacy-preserving data mining,'' in \emph{Proc.
  ACM SIGMOD International Conference on Management of Data (SIGMOD)}.\hskip
  1em plus 0.5em minus 0.4em\relax Dallas, TX: Association for Computing
  Machinery, May 16-18 2000, pp. 439--450.

\bibitem{dwork2006calibrating}
C.~Dwork, F.~McSherry, K.~Nissim, and A.~Smith, ``Calibrating noise to
  sensitivity in private data analysis,'' in \emph{Proc. Third Theory of
  Cryptography Conference (TCC)}.\hskip 1em plus 0.5em minus 0.4em\relax New
  York, NY: Springer, March 4-7 2006, pp. 265--284.

\bibitem{10.1007/11787006_1}
C.~Dwork, ``Differential privacy,'' in \emph{Automata, Languages and
  Programming}, M.~Bugliesi, B.~Preneel, V.~Sassone, and I.~Wegener, Eds.\hskip
  1em plus 0.5em minus 0.4em\relax Berlin, Heidelberg: Springer Berlin
  Heidelberg, 2006, pp. 1--12.

\bibitem{6482222}
L.~Sankar, S.~R. Rajagopalan, and H.~V. Poor, ``Utility-privacy tradeoffs in
  databases: An information-theoretic approach,'' \emph{IEEE Transactions on
  Information Forensics and Security}, vol.~8, no.~6, pp. 838--852, June 2013.

\bibitem{971193}
P.~Samarati, ``Protecting respondents identities in microdata release,''
  \emph{IEEE Transactions on Knowledge and Data Engineering}, vol.~13, no.~6,
  pp. 1010--1027, Nov./Dec. 2001.

\bibitem{10.1142/S0218488502001648}
L.~Sweeney, ``$k$-anonymity: A model for protecting privacy,'' \emph{Int. J.
  Uncertain. Fuzziness Knowl.-Based Syst.}, vol.~10, no.~5, pp. 557--570, Oct.
  2002.

\bibitem{10.1145/1217299.1217302}
A.~Machanavajjhala, D.~Kifer, J.~Gehrke, and M.~Venkitasubramaniam,
  ``$\ell$-diversity: Privacy beyond $k$-anonymity,'' \emph{ACM Trans. Knowl.
  Discov. Data}, vol.~1, no.~1, pp. 3--es, Mar. 2007.

\bibitem{4221659}
N.~Li, T.~Li, and S.~Venkatasubramanian, ``$t$-closeness: Privacy beyond
  $k$-anonymity and $\ell$-diversity,'' in \emph{Proc. 2007 IEEE 23rd
  International Conference on Data Engineering}, Istanbul, Turkey, 15-20 Apr.
  2007, pp. 106--115.

\bibitem{10.1007/s00778-014-0351-4}
J.~Soria-Comas, J.~Domingo-Ferrer, D.~S\'{a}nchez, and S.~Mart\'{\i}nez,
  ``Enhancing data utility in differential privacy via microaggregation-based
  k-anonymity,'' \emph{The VLDB Journal}, vol.~23, no.~5, p. 771–794, Oct.
  2014.

\bibitem{6970882}
A.~Makhdoumi, S.~Salamatian, N.~Fawaz, and M.~Médard, ``From the information
  bottleneck to the privacy funnel,'' in \emph{Proc. IEEE Information Theory
  Workshop (ITW)}, Hobart, TAS, Australia, Nov. 2-5, 2014, pp. 501--505.

\bibitem{DBLP:journals/corr/ZhangOSO17}
Y.~Zhang, M.~Ozay, Z.~Sun, and T.~Okatani, ``Information potential
  auto-encoders,'' \emph{CoRR}, vol. abs/1706.04635, 2017.

\bibitem{DBLP:conf/aaai/XuCSMS17}
K.~Xu, T.~Cao, S.~Shah, C.~Maung, and H.~Schweitzer, ``Cleaning the null space:
  {A} privacy mechanism for predictors,'' in \emph{Proc. Thirty-First {AAAI}
  Conference on Artificial Intelligence}, S.~P. Singh and S.~Markovitch,
  Eds.\hskip 1em plus 0.5em minus 0.4em\relax San Francisco, CA: {AAAI} Press,
  Feb. 2017, pp. 2789--2795.

\bibitem{Belghazi2018MutualIN}
M.~I. Belghazi, A.~Baratin, S.~Rajeswar, S.~Ozair, Y.~Bengio, R.~D. Hjelm, and
  A.~C. Courville, ``{MINE}: Mutual information neural estimation,'' in
  \emph{Proc. 35th International Conference on Machine Learning (ICML)},
  vol.~80.\hskip 1em plus 0.5em minus 0.4em\relax Stockholmsmässan, Sweden:
  PMLR, Jul. 10-15, 2018, pp. 531--540.

\bibitem{Fraser1986IndependentCF}
Fraser and Swinney, ``Independent coordinates for strange attractors from
  mutual information,'' \emph{Phys. Rev. A: General physics}, vol.~33, no.~2,
  pp. 1134--1140, Feb. 1986.

\bibitem{761290}
G.~Darbellay and I.~Vajda, ``Estimation of the information by an adaptive
  partitioning of the observation space,'' \emph{IEEE Transactions on
  Information Theory}, vol.~45, no.~4, pp. 1315--1321, May 1999.

\bibitem{PhysRevE.69.066138}
A.~Kraskov, H.~St\"ogbauer, and P.~Grassberger, ``Estimating mutual
  information,'' \emph{Phys. Rev. E}, vol.~69, p. 066138, Jun. 2004.

\bibitem{8294268}
W.~Gao, S.~Oh, and P.~Viswanath, ``Demystifying fixed $k$ -nearest neighbor
  information estimators,'' \emph{IEEE Transactions on Information Theory},
  vol.~64, no.~8, pp. 5629--5661, Feb. 2018.

\bibitem{DBLP:journals/jmlr/SuzukiSSK08}
T.~Suzuki, M.~Sugiyama, J.~Sese, and T.~Kanamori, ``Approximating mutual
  information by maximum likelihood density ratio estimation,'' in \emph{Proc.
  Third Workshop on New Challenges for Feature Selection in Data Mining and
  Knowledge Discovery (FSDM)}, Y.~Saeys, H.~Liu, I.~Inza, L.~Wehenkel, and
  Y.~V. de~Peer, Eds.\hskip 1em plus 0.5em minus 0.4em\relax Antwerp, Belgium:
  JMLR.org, Sept. 15, 2008, pp. 5--20.

\bibitem{Barber2003TheIA}
D.~Barber and F.~Agakov, ``The {IM} algorithm: a variational approach to
  information maximization,'' in \emph{Proc. Neural Information Processing
  Systems: Natural and Synthetic (NIPS)}, Vancouver, Canada, Dec. 8-13, 2003.

\bibitem{9022174}
A.~Elad, D.~Haviv, Y.~Blau, and T.~Michaeli, ``Direct validation of the
  information bottleneck principle for deep nets,'' in \emph{Proc. IEEE/CVF
  International Conference on Computer Vision Workshop (ICCVW)}, Seoul, South
  Korea, Oct. 27-28, 2019, pp. 758--762.

\bibitem{Qian2021LearningRV}
W.~Qian, B.~Chen, and X.~Huang, ``Learning robust variational information
  bottleneck with reference,'' \emph{ArXiv}, vol. abs/2104.14379, Apr. 2021.

\bibitem{hjelm2019learning}
D.~Hjelm, A.~Fedorov, S.~Lavoie-Marchildon, K.~Grewal, P.~Bachman,
  A.~Trischler, and Y.~Bengio, ``Learning deep representations by mutual
  information estimation and maximization,'' in \emph{Proc. International
  Conference on Learning Representations (ICLR)}.\hskip 1em plus 0.5em minus
  0.4em\relax ICLR, Apr. 2019.

\bibitem{DBLP:conf/allerton/CalmonF12}
F.~du~Pin~Calmon and N.~Fawaz, ``Privacy against statistical inference,'' in
  \emph{Proc. 50th Annual Allerton Conference on Communication, Control, and
  Computing, Allerton 2012}.\hskip 1em plus 0.5em minus 0.4em\relax Allerton
  Park {\&} Retreat Center, Monticello, IL: {IEEE}, Oct. 1-5, 2012, pp.
  1401--1408.

\bibitem{book}
J.~Domingo-Ferrer, D.~Sánchez, and J.~Soria-Comas, \emph{Database
  Anonymization: Privacy Models, Data Utility, and Microaggregation-based
  Inter-model Connections}, ser. Synthesis Lectures on Information Security,
  Privacy, and Trust.\hskip 1em plus 0.5em minus 0.4em\relax San Rafael, CA:
  Morgan \& Claypool, 2016, vol.~8.

\bibitem{https://doi.org/10.1002/cpa.3160360204}
M.~D. Donsker and S.~R.~S. Varadhan, ``Asymptotic evaluation of certain markov
  process expectations for large time. {IV},'' \emph{Communications on Pure and
  Applied Mathematics}, vol.~36, no.~2, pp. 183--212, Mar. 1983.

\bibitem{8815464}
R.~Fritschek, R.~F. Schaefer, and G.~Wunder, ``Deep learning for channel coding
  via neural mutual information estimation,'' in \emph{2019 IEEE 20th
  International Workshop on Signal Processing Advances in Wireless
  Communications (SPAWC)}, 2-5 July 2019, pp. 1--5.

\end{thebibliography}
\end{document}